\documentclass[a4paper,onecolumn,11pt,accepted=2022-02-02]{quantumarticle}
\pdfoutput=1
\usepackage[utf8]{inputenc}
\usepackage[a4paper, total={6in, 8in}]{geometry}
\usepackage{amsmath}
\usepackage{amssymb}
\usepackage{amsbsy}
\usepackage{amsthm}
\usepackage{mathtools}
\usepackage{xcolor}
\usepackage{thmtools}
\usepackage{thm-restate}
\usepackage{hyperref}
\hypersetup{ colorlinks = true, allcolors = {}}
\usepackage{cleveref}

\usepackage[numbers, sort&compress]{natbib}

\makeatletter
\newcommand{\dotDelta}{{\vphantom{\Delta}\mathpalette\d@tD@lta\relax}}
\newcommand{\d@tD@lta}[2]{%
	\ooalign{\hidewidth$\m@th#1\mkern-1mu\cdot$\hidewidth\cr$\m@th#1\Delta$\cr}%
}
\makeatother

\newcommand{\F}{\mathbb{F}}
\newcommand{\C}{\mathbb{C}}
\newcommand{\T}{\mathbb{T}}

\newcommand{\R}{\mathbb{R}}
\newcommand{\Z}{\mathbb{Z}}

\newcommand{\ket}[1]{|#1\rangle}
\newcommand{\bra}[1]{\langle#1|}
\newcommand{\ketbra}[2]{|#1\rangle\langle#2|}
\newcommand{\braket}[2]{\langle#1\vert#2\rangle}

\newcommand{\beq}{\begin{equation}}
	\newcommand{\eeq}{\end{equation}}
\newcommand{\beqn}{\begin{equation*}}
	\newcommand{\eeqn}{\end{equation*}}
\newcommand{\beqr}{\begin{eqnarray}}
	\newcommand{\eeqr}{\end{eqnarray}}
\newcommand{\beqrn}{\begin{eqnarray*}}
	\newcommand{\eeqrn}{\end{eqnarray*}}
\newcommand{\bmline}{\begin{multline}}
	\newcommand{\emline}{\end{multline}}
\newcommand{\bmlinen}{\begin{multline*}}
	\newcommand{\emlinen}{\end{multline*}}

\newtheorem{theorem}{Theorem}[section]
\newtheorem{lemma}[theorem]{Lemma}
\newtheorem{claim}[theorem]{Claim}
\newtheorem{proposition}[theorem]{Proposition}

\newtheorem{definition}[theorem]{Definition}

\newtheorem{problem}[theorem]{Problem}
\newtheorem{example}[theorem]{Example}

\DeclareMathOperator*{\E}{\mathbb{E}}

\renewcommand{\ell}{L}

\title{Stabilizer rank and higher-order Fourier analysis}

\begin{document}

	\author{Farrokh Labib}
	\affiliation{CWI, QuSoft, Science Park 123, 1098 XG Amsterdam, Netherlands. Supported by the Gravitation-grant NETWORKS-024.002.003 from the Dutch Research Council~(NWO).}
	\email{farrokhlabib@gmail.com}
	
		\maketitle
\begin{abstract}
		We establish a link between stabilizer states, stabilizer rank and higher-order Fourier analysis---a still-developing area of mathematics that grew out of Gowers's celebrated Fourier-analytic proof of Szemer\'edi's theorem \cite{gowers1998new}. We observe that $n$-qudit stabilizer states are so-called \emph{nonclassical quadratic phase functions} (defined on affine subspaces of $\F_p^n$ where $p$ is the dimension of the qudit) which are fundamental objects in higher-order Fourier analysis. This allows us to import tools from this theory to analyze the stabilizer rank of quantum states. Quite recently, in \cite{peleg2021lower} it was shown that the $n$-qubit magic state has stabilizer rank $\Omega(n)$. Here we show that the qudit analogue of the $n$-qubit magic state has stabilizer rank $\Omega(n)$, generalizing their result to qudits of any prime dimension. Our proof techniques use explicitly tools from higher-order Fourier analysis.		
		We believe this example motivates the further exploration of applications of higher-order Fourier analysis in quantum information theory.

\end{abstract}

\section{Introduction}


The Gottesman-Knill Theorem \cite{gottesman1998heisenberg, nielsen2002quantum} states that any quantum circuit consisting of Clifford gates can be efficiently classically simulated. The Clifford group on $n$ qubits is generated by the Hadamard gate $H$, the $\pi/4$ phase gate $S$ and the entangling $\mathrm{CNOT}$ gate. In particular, this means that circuits consisting only of Clifford gates cannot provide computational advantage over classical computers. We can promote such circuits to universal quantum computers by having access to a non-Clifford gate or (equivalently) a ``magic state'' \cite{bravyi2005universal}. It is widely believed that universal quantum computers cannot be efficiently simulated by classical computers: state-of-the-art simulators using modern day supercomputers are only able to simulate a few dozens of qubits \cite{chen2018classical,haner20175,pednault2017breaking,smelyanskiy2016qhipster}. So it has to be this magic state that fuels the computational hardness of simulation by classical computers. It is therefore important to understand how much this resource costs in terms of free (efficiently simulatable) resources.  These costs are quantified by  ``measures of magic'' \cite{liu2020many} an example of which is \emph{stabilizer rank}, first introduced in \cite{bravyi2016trading}. Here, the free resources are states obtained from the canonical all-zero state by applying only Clifford operations, which are the well-known stabilizer states. To increase our understanding of non-stabilizerness, or the amount of ``magic'' a quantum state has, a valid approach might be to find a different characterization of these objects. This might introduce new techniques in analyzing measures of magic. It is well known that stabilizer states are quadratic forms defined on affine subspaces \cite{dehaene2003clifford,hostens2005stabilizer, gross2006hudson}. Here we observe that these objects are so-called \emph{nonclassical quadratic phase functions} defined on affine subspaces which are well-studied objects in higher-order Fourier analysis.\\
\\
Higher-order Fourier analysis, a still nascent area of mathematics, grew out of a Fourier-analytic proof of Szemer\'edi's Theorem \cite{szemeredi1975sets} by Gowers \cite{gowers1998new}.
Whereas in Fourier analysis one studies how functions correlate with characters, in higher-order Fourier analysis one studies correlations with functions that resemble polynomials (where characters correspond to linear functions). These polynomial-like functions are known as \emph{polynomial phase functions}.

It turns out that Boolean functions giving the (unnormalised) amplitudes of graph states are examples of quadratic phase functions. In general, quadratic phase functions can be defined using ``multiplicative derivatives'': for $h\in \F_2^n$, the multiplicative derivative of $f\colon \F_2^n\to\C$ in direction $h$ is
\begin{align*}
	\dotDelta_hf(x):=f(x+h)\overline{f(x)}.
\end{align*} 
\emph{Nonclassical} polynomial phase functions of degree $d$ are those functions that are constant after taking $d+1$ multiplicative derivatives. It is not difficult to check that graph states, whose amplitude function always has the form $f(x)=(-1)^{q(x)}$ where $q$ is a quadratic polynomial, satisfies this property with $d=2$. These are referred to as the \emph{classical} quadratic phase functions. However, the nonclassical quadratic phase functions are not exhausted by these examples. It turns out that stabilizer states correspond to functions in this broader class. This establishes a surprising link between higher-order Fourier analysis and quantum information theory. It was shown in \cite{dehaene2003clifford} (see also \cite{bravyi2016improved}) that stabilizer states are \emph{quadratic forms} taking values in~$\Z_8$ defined on affine subspaces. We will see that they are nonclassical quadratic phase function on affine subspaces.

Let $p$ be an odd prime and consider qudits of dimension $p$. Then, the amplitudes of $n$-qudit stabilizer states are also quadratic phase functions defined on affine subspaces of $\F_p^n$ \cite{hostens2005stabilizer}, see also \cite{gross2006hudson}. It is interesting to note that for primes $p>2$ the $n$-qudit stabilizer states are given by classical quadratic polynomials while there are no nonclassical polynomials of degree two, contrary to the case $p=2$.
%

\paragraph{Stabilizer rank.} \emph{Stabilizer rank} is a measure of magic which was recently extensively analysed by Bravyi et al.\ \cite{bravyi2019simulation}. The stabilizer rank of a quantum state $\ket{\psi}$, denoted $\chi(\ket{\psi})$, is the minimal number $r$ such that $\ket{\psi}$ can be written as a linear combination of $r$ stabilizer states. As is well known, any circuit $\mathcal{C}$ consisting of Clifford gates and $n$ copies of the $T$-gate, given by $\ket{0}\!\bra{0}+e^{i\pi/4}\ket{1}\!\bra{1}$, can be implemented using Clifford operations on the $n$-qubit magic state~$\ket{T}^{\otimes n}$, where $\ket{T}=\frac{\ket{0}+e^{i\pi/4}\ket{1}}{\sqrt{2}}$. Then, the stabilizer rank of $\ket{T}^{\otimes n}$ upper bounds the simulation cost of the circuit~$\mathcal{C}$.  Bravyi, Smith and Smolin \cite{bravyi2016trading} showed that the stabilizer rank of the $n$-qubit magic state is $\Omega(\sqrt{n})$. 
Very recently Peleg, Shpilka and Volk \cite{peleg2021lower} showed a lower bound of $\Omega(n)$ for the stabilizer rank of the $n$-qubit magic state, which is a quadratic improvement.
Here we give the same lower bound but generalize to qudits of any prime dimension. 

Adding any non-Clifford gate to the Clifford gate set could in principle promote it to universal quantum computation. However, we use the generalization of the $T$ gate for qudits from \cite{howard2012qudit}. Let us call this gate $U$ and define $\ket{+}: = \frac{\ket{0}+\ket{1}+\dots +\ket{p-1}}{\sqrt{p}}$. Then, the single qudit magic state over $\F_p$ is defined to be
\begin{align*}
	\ket{\psi_U} = U\ket{+}.
\end{align*}
Our main result is the following.
\begin{theorem}\label{thm:main}
	Let $p$ be any prime and let $\ket{\psi_U}^{\otimes n}$ be the $n$-qudit magic state over $\F_p$. We have that $\chi(\ket{\psi_U}^{\otimes n})\geq \Omega(n)$.
\end{theorem}
This result generalizes \cite{peleg2021lower}, but our techniques are completely different and use  explicitly tools from higher-order Fourier analysis. Roughly speaking, we show that the function giving the amplitudes of the $n$-qudit magic state is a cubic nonclassical polynomial phase function for which the polynomial ``in the phase'' has high rank (see Definition \ref{def:rank_poly}). We then prove that the lower bound for this rank is also a lower bound for the stabilizer rank. In this step we use a lemma from \cite{peleg2021lower} (Claim~3.3 in that paper) to get a handle on the affine subspaces that appear from the stabilizer states. Apart from this lemma, the techniques are different.

We think that the techniques used here might pave the way to super-linear lower bounds for decompositions in terms of stabilizer states defined on the full space $\F_p^n$.

\paragraph{Higher-order Fourier analysis and quantum information theory.} 
We believe that the application we found motivate the further study of tools in higher-order Fourier analysis in quantum information theory. A work that connected to this area is \cite{bannink2019bounding}, where the ``Gowers Inverse Theorem'' is utilised to analyze certain Bell inequalities.
This current work is yet another example.

Higher-order Fourier analysis has already found many applications in classical theoretical computer science, such as in property testing, coding theory and complexity theory \cite{hatami2019higher}. In analogy with analysis of Boolean functions, we hope that higher-order Fourier analysis proves equally useful in quantum information theory.

\paragraph{Acknowledgements} I would like to thank Jop Bri\"et for his guidance and many fruitful discussions in this project. Also thanks to Michael Walter for the fruitful discussions on stabilizer and magic states in odd prime dimension.

\section{Preliminaries}
For a finite non-empty set $S$, we write $\E_{s\in S}$ to mean $\frac{1}{|S|}\sum_{s\in S}$. Let $p$ be a prime. We denote by $\F_p$ the field of $p$ elements, $\T=\R/\Z$ and $\Z_n$ is the ring of integers mod $n$.

Let $|\cdot|\colon \F_p\to \{0,1,\dots, p-1\}$ be the natural map. This map has the property that for~$a,b\in \F_2$,
\begin{align}\label{eq:property_||}
	|a+b|=|a|+|b|-2|ab|.
\end{align}
And for $a,b\in \F_3$,
\begin{align}\label{eq:property_||_3}
	|a+b|=|a|+|b|+3(|a^2b|+|ab^2|)-6|ab|\mod 9.
\end{align}
We also define $|\cdot|\colon \F_p^n\to\Z_{\geq 0}:x\mapsto |x_1|+\dots +|x_n|$, this is an abuse of notation, but it should be clear from context what the domain of the map $|\cdot|$ is.
Define $\iota$ to be the map given by
\begin{align*}
	\iota \colon \F_p\to \T\colon  x\mapsto |x|/p\mod 1.
\end{align*}
The exponential map $e\colon \T\to \C$ is defined $e(t):=e^{2\pi i t}$.

For an affine subspace $H\subset \F_p^n$, we write $H=L(H)+v$ for some vector $v\in \F_p^n$ and a linear subspace $L(H)$. The vector $v$ is not unique, but $L(H)$ is uniquely determined by $H$, namely~$L(H)=\{x-y\colon x,y\in H\}$. 

Let $G$ be an abelian group, $H\subset \F_p^n$ an affine subspace and $P\colon H\to G$ a map. The additive derivative $\Delta_h$ is defined to be $\Delta_hP(x):=P(x+h)-P(x)$, where $h\in L(H)$ and $x\in H$.

For a subset $A\subset \F_p^n$, write $1_A\colon \F_p^n\to\{0,1\}$ for the indicator function of $A$, i.e. $1_A(x)=1$ if $x\in A$ and 0 otherwise.

We also need some basic definitions for the Fourier transform. Let $\omega_p:=e^{2\pi i/p}$. For functions~$f,g\colon \F_p^n\to \C$, their inner product is defined by
\begin{align*}
	\langle f,g\rangle :=\E_{x\in\F_p^n}f(x)\overline{g(x)}=p^{-n}\sum_{x\in \F_p^n}f(x)\overline{g(x)}.
\end{align*}
Denote by $\hat{f}$ the Fourier transform of $f$, i.e.
\begin{align*}
	\hat{f}(\alpha) = \E_{x\in\F_p^n} f(x)\omega_p^{\langle \alpha, x\rangle},\quad \alpha \in\F_p^n.
\end{align*}
The inner product for the Fourier transforms is defined by
\begin{align*}
	\langle \hat{f},\hat{g}\rangle =\sum_{\alpha\in\F_p^n}\hat{f}(\alpha)\overline{\hat{g}(\alpha)}.
\end{align*}
Using this definition, Parseval's theorem becomes
\begin{align*}
	\langle f,g \rangle = \langle \hat{f},\hat{g}\rangle.
\end{align*}

\section{Techniques}
In this section, we introduce all the definitions necessary for the proof of our main result. Our main result generalizes the recent result of \cite{peleg2021lower}, but we use completely different techniques that we explain here as well.
\subsection{Stabilizer states}
In \cite{bravyi2016improved}, a succinct representation of an $n$-qubit stabilizer states is given in terms of quadratic forms on affine subspaces of $\F_2^n$. For this, they introduced the following definition.

\begin{definition}\label{def:quadratic-form}
	For an affine subspace $H\subset \F_2^n$, a map $Q\colon H\to \frac{1}{8}\Z/\Z$ is called a quadratic form if $\Delta_{h_1}\Delta_{h_2}Q(x)$ is independent of $x\in H$ for all $h_1,h_2\in L(H)$.
\end{definition}
\begin{theorem}[\cite{bravyi2016improved}]\label{thm:bravyi2016}
	For any $n$-qubit stabilizer state $\ket{\phi}$, there exists a unique affine subspace~$H\subset \F_2^n$ and a quadratic form $Q\colon H\to \frac{1}{8}\Z/\Z$ such that
	\begin{align}\label{stab_states_quadraticform}
		\ket{\phi} =2^{-\dim{H}/2} \sum_{x\in H}e(Q(x))\ket{x}.
	\end{align}
\end{theorem}
In \cite{bravyi2016improved} quadratic forms are considered that are maps $Q\colon H\to\frac{1}{8}\Z/\Z$. We will show that the only way that such a map has the property $\Delta_{h_1}\Delta_{h_2}Q(x)$ being independent of $x\in H$ for all $h_1,h_2\in L(H)$, is if $Q$ actually take values in $\frac{1}{4}\Z/\Z\subset\frac{1}{8}\Z/\Z$.  
We will also see that such functions are known as \emph{nonclassical} polynomials of degree two in the literature of higher-order Fourier analysis. Making this  explicit link with higher-order Fourier analysis allows us to import tools from that theory to use in lower bounding the stabilizer rank of quantum states.

Next, qudit stabilizer states where the dimension of the qudit is an odd prime $p$ are somewhat simpler as they are given by quadratic polynomials taking values in $\F_p$ on affine subspaces in~$\F_p^n$. Let $\omega_p= e^{2\pi i/p}$ be a $p$-th root of unity. 
A map $Q\colon H\to\F_p$ is a quadratic polynomial on an affine subspace~$H\subset \F_p^n$ if $\Delta_{h_1}\Delta_{h_2}Q(x)$ is independent of $x\in H$ for all $h_1,h_2\in L(H)$.
\begin{theorem}[\cite{hostens2005stabilizer}, see also \cite{gross2006hudson}]\label{thm:stabstates_quadraticPolys}
	Let $p$ be an odd prime and $\ket{\phi}$ an $n$-qudit stabilizer state where the dimension of the qudit is $p$. Then, there is an affine subspace $H\subset \F_p^n$ and a quadratic polynomial $Q\colon H\to \F_p$ such that 
	\begin{equation}
		\ket{\phi} = p^{-\dim(H)/2}\sum_{x\in H}\omega_p^{Q(x)}\ket{x}.
	\end{equation}
\end{theorem}
\begin{definition}\label{def:stab}
	For an $n$-qudit quantum state $\ket{\psi}$ define its stabilizer rank, denoted $\chi(\ket{\psi})$, to be the minimal number $r$ needed to write $\ket{\psi}$ as a linear combination of $r$ stabilizer states.
\end{definition}

\subsection{Nonclassical polynomials}
We will now define what \emph{nonclassical} polynomials are. 
First, let $P\in\F_p[x_1,\dots,x_n]$ be a ``classical'' polynomial of degree $d\geq 1$. One can show that $\Delta_hP(x)=P(x+h)-P(x)$ is again a polynomial but of degree at most $d-1$. So after taking $d+1$ derivatives, the resulting polynomial will be the zero polynomial. Using this observation, we can define a broader class of polynomials as functions that take values in $\T$ and satisfy a condition on its derivatives.

\begin{definition}\label{def:non-classical_polynomial}
	For an integer $d\geq 1$, a map $P\colon \F_p^n\to \T$ is a nonclassical polynomial of degree at most $d$ if for all $h_1,\dots, h_{d+1}\in \F_p^n$ we have
	\begin{align}
		\Delta_{h_{d+1}}\cdots \Delta_{h_1}P(x)=0.
	\end{align}
	The degree of $P$ is the smallest such $d$. 
\end{definition}
Using the map $\iota\colon \F_p\to\T\colon x\to |x|/p$, one can view a polynomial $P\in\F_p[x_1,\dots,x_n]$ as a map~$\iota(P)\colon \F_p^n\to \T$. Nonclassical polynomials that arise in this way are called classical polynomials and they are a subset of the nonclassical polynomials. Note that they take values in~$\frac{1}{p}\Z/\Z$.
This confusing terminology has unfortunately become standard in the literature.
Below, all polynomials are assumed to be nonclassical if it is not explicitly stated. The following example shows that this containment of classical polynomials in the set of nonclassical polynomials is indeed proper.
\begin{example}\label{example:cubic_polynomials}
	\normalfont
	Let $p=2$ and consider $P\colon \F_2^n\to\T$ be given by $x\mapsto |x|/4$. This is a nonclassical polynomial of degree two. To show this, we take derivatives:
	\begin{align*}
		|x+h|/4-|x|/4 = |x|/4+|h|/4-|x\circ h|/2-|x|/4=|h|/4-|x\circ h|/2\mod 1,
	\end{align*}
	where $\circ$ is entry-wise product of vectors. In the first equality we used the property in (\ref{eq:property_||}). Taking one more derivative
	\begin{align*}
		\Delta_{h'}\Delta_hP(x)=-|x\circ h +h'\circ h|/2+|x\circ h|/2=-|h'\circ h|/2\mod 1.
	\end{align*}
	Indeed, $P(x)=|x|/4$ is a nonclassical polynomial of degree two. Note that it is \emph{not} a classical polynomial since it takes values in $\frac{1}{4}\Z/\Z$. \newline
	In general, the polynomial $P(x)=|x|/2^k$ is a nonclassical polynomial of degree $k$. We will see later (Section \ref{subsec:magic_states}) that for $k=3$, this polynomial corresponds to the $n$-qubit magic state.

	Let $p=3$ and consider the polynomial $P\colon\F_3^n\to\T\colon x\to |x|/9$. This is a nonclassical polynomial of degree three: let $h\in \F_3^n$, then by the property in (\ref{eq:property_||_3})
	\begin{align*}
		\Delta_hP(x)=|x+h|/9-|x|/9 = |h|/9+(|x^2\circ h|+|x\circ h^2|)/3-2|x\circ h|/3\mod 1.
	\end{align*}
	Here $x^2 = x\circ x$, i.e. the entry-wise product of $x$ with itself. The above is (up to a shift) a classical polynomial of degree two, hence $P$ is cubic. We will see in Section \ref{subsec:magic_states} that this polynomial $P$ corresponds to the $n$-qudit magic state where the qudit has dimension three. 
\end{example}

The above is a local definition of nonclassical polynomials. A global definition of classical polynomials of degree at most $d$ is that they take the form $\sum_{i_1+\dots +i_n\leq d}c_{i_1,\dots, i_n}x_1^{i_1}\cdots x_n^{i_n}$. Similarly, nonclassical polynomials have the following global description.

\begin{proposition}[\cite{tao2012inverse}]\label{prop:global_nonclassical_poly}
	$P\colon \F_p^n\to \T$ is a nonclassical polynomial of degree at most $d$ if and only if it has a representation of the form
	\begin{align}
		P(x_1,\dots,x_n)=\alpha+\sum_{\substack{ 0\leq i_1,\dots i_n\leq p-1;j\geq 0: \\ 0<i_1+\dots + i_n\leq d-j(p-1)}}\frac{c_{i_1,\dots,i_n,j}|x_1|^{i_1}\cdots |x_n|^{i_n}}{p^{j+1}},
	\end{align}
	for some coefficients $c_{i_1,\dots,i_n,j}\in \{0,1,\dots, p-1\}$ and $\alpha\in\T$, these coefficients are unique. The maximal $j$ in this decomposition is called the \emph{depth} of $P$.
\end{proposition}
Note that the depth of a nonclassical polynomial is always at most $\lceil d/(p-1)\rceil -1$ since at least one of the indices $i_1,\dots,i_n$ must be positive.
Using this proposition, one quickly sees that the polynomial $x\mapsto |x|/4$ for $x\in\F_2^n$ from the example above has degree two and depth 1. The polynomial $x\mapsto |x|/2^{k}$ has degree $k$ and depth $k-1$.
\newline
\\
Let $H\subset \F_p^n$ be an affine subspace and $Q\colon H\to \T$ a map. Let $A\colon\F_p^{\dim(H)}\to\F_p^n$ be an affine linear transformation such that $A$ maps surjectively on $H$. We can view $Q$ in its ``local coordinate system'': for $y\in \F_p^{\dim(H)}$ define $Q'(y)=Q(A(y))$. 

Let $p=2$. If $Q$ is a quadratic form (Definition \ref{def:quadratic-form}), then $Q'$ is a nonclassical polynomial of degree at most two. Proposition \ref{prop:global_nonclassical_poly} implies that its depth is at most 1. In other words, $Q'$ takes values in $\frac{1}{4}\Z/\Z$ and so does $Q$.

Similarly, for $p>2$ a quadratic polynomial $Q\colon H\to\F_p$ on an affine subspace $H\subset \F_p^n$, is simply a quadratic polynomial in local coordinates, i.e. the map $Q'\colon \F_p^{\dim(H)}\to \F_p$ is a polynomial of degree two in the $\dim(H)$ variables.


The following is a simple observation that we will need later, it follows directly from Proposition \ref{prop:global_nonclassical_poly}.
\begin{lemma}\label{lem:nonclassical_degree_1_are_classical}
	Let $P\colon\F_p^n\to\T$ be a nonclassical polynomial of degree at most one. Then $P$ is classical (up to some shift).
\end{lemma}

\subsection{Rank of nonclassical polynomials}
In this section we introduce a notion of rank for nonclassical polynomials. The main reference is again \cite{hatami2019higher}.
\begin{definition}\label{def:rank_poly}
	Let $P\colon \F_p^n\to\T$ be a nonclassical polynomial. For an integer $d\geq 1$, we define the $d$-rank, denoted by $\mathrm{rank}_d(P)$, to be the minimal number $r$ such that there are nonclassical polynomials $Q_1,\dots, Q_r$ all of degree at most $d-1$ and a function $\Gamma\colon \T^r\to \T$ such that for all $x\in \F_p^n$
	\begin{align*}
		\Gamma(Q_1(x),\dots, Q_r(x))=P(x).
	\end{align*}
	If $d=1$, the $d$-rank will be $\infty$ if $P$ is non-constant and 0 otherwise. The rank of  $P$, denoted $\mathrm{rank}(P)$, is the $\deg(P)$-rank of $P$.
\end{definition}
The next result \cite[Corollary 7.16]{hatami2019higher} is a standard application of Fourier analysis and can be thought of as a kind of ``inverse theorem''. We give the proof here for convenience.
\begin{proposition}\label{prop:rank_implies_uniform}
	Let $d\geq 2$ and $P\colon \F_p^n\to \T$ a polynomial with $r:=\mathrm{rank}_d(P)$. Then there exists a polynomial $Q$ of degree at most $d-1$ such that 
	\begin{equation*}
		|\langle e(P),e(Q)\rangle |\geq p^{-(1+\lceil(d-1)/(p-1)\rceil)r}.
	\end{equation*}
\end{proposition}
\begin{proof}
	Let $\Gamma\colon \T^r\to\T$ be a map and $Q_1,\dots,Q_r$ be polynomials of degree at most $d-1$ such that
	\begin{align*}
		P(x)=\Gamma(Q_1(x),\dots,Q_r(x)).
	\end{align*}
We can assume that $\Gamma$ is a map with domain $G:=\prod_{i=1}^r\frac{1}{p^{k_i+1}}\Z/\Z$, where $k_i$ is the depth of $Q_i$. Let $\hat{G}=\prod_{i=1}^r\Z_{p^{k_i+1}}$ be the dual of $G$, so that the Fourier decomposition of $e(\Gamma)$ becomes
\begin{align*}
	e(\Gamma(z))=\sum_{\alpha\in \hat{G}}\hat{\Gamma}_\alpha e(\langle \alpha,z\rangle).
\end{align*}
The Fourier decomposition of $e(\Gamma)$ gives a ``higher-order'' Fourier decomposition of $e(P)$: for~$\alpha\in~\hat{G}$ define $Q_\alpha(x):=\sum_{i=1}^r\alpha_iQ_i(x)$, then
\begin{align*}
	e(P(x)) = \sum_{\alpha\in\hat{G}}\hat{\Gamma}_\alpha e(Q_\alpha(x)).
\end{align*}
So
\begin{align*}
	1=|\langle e(P),e(P)\rangle|\leq \sum_{\alpha\in \hat{G}}|\langle e(P),e(Q_\alpha) \rangle|.
\end{align*}
Indeed, there is an $\alpha^*$ such that
\begin{align*}
	|\langle e(P),e(Q_{\alpha^*}) \rangle|\geq |\hat{G}|^{-1}.
\end{align*}
Since the degree of the polynomials $Q_i$ is at most $d-1$, it follows that $k_i(p-1)\leq d-1$. This implies that 
\begin{align*}
	|\hat{G}| = p^{k_1+1+\dots+k_r+1} \leq p^{(1+\lceil(d-1)/(p-1)\rceil)r}.
\end{align*}
\end{proof}

The next lemma \cite[Lemma 7.2]{hatami2019higher} tells us how the rank changes if we restrict a polynomial to an affine subspace.
\begin{lemma}\label{lem:restrict_pol_subspace]}
	Let $P\colon \F_p^n\to\T$ be a polynomial of degree $d\geq 2$ and $r:=\mathrm{rank}(P)$. Let $U\subset \F_p^n$ be an affine subspace of codimension $k$ and define $P'$ to be the restriction of $P$ to~$U$. If $r>pk+1$ then $P'$ is a polynomial of degree $d$ and $\mathrm{rank}(P')\geq r-pk$.
\end{lemma}
\begin{proof}
	We will prove the statement for $k=1$. The general case will follow after repeated application of the proof for $k=1$.
	
%
	Since rank and degree do not change under invertible affine linear transformations, we can assume without loss of generality that $U=\{x\in \F_p^n\colon x_n=0\}$. Let $\pi\colon \F_p^n\to \F_p^n$ be the projection onto $U$, so $\pi(x_1,\dots,x_n)=(x_1,\dots,x_{n-1},0)$. Define $P''=P-P'\circ \pi$. For $x\in U$ we have~$P''(x)=0$. Let $a\in\F_p\setminus \{0\}$ and $h_a=(0,\dots,0,a)\in \F_p^n$. We have that $\Delta_{h_a}P''$ has degree at most~$d-1$ and that~$\Delta_{h_a}P''(y)=P''(y+h)$ for all $y\in U$. So $P''$ agrees with a polynomial~$Q_a$ of degree at most $d-1$ on $U+h_a$. This implies there is a function $\Gamma\colon \T^{p+1}\to\T$ such that~$P(x)=\Gamma(|x_n|/p,P'(x),Q_1(x),\dots, Q_{p-1}(x))$. 
	
	Now, if $P'$ has degree at most $d-1$, then $\mathrm{rank}(P)\leq p+1<r$, which is a contradiction. If~$P'$ has rank $<r-p$, we get that $\mathrm{rank}(P)<r$ which is again a contradiction.
\end{proof}

Next, we define the \emph{Fourier rank} of a function which is a new notion of rank for functions and prove straightforward properties of it.
\begin{definition}\label{def:frank}
	Let $f\colon \F_p^n\to \C$ be a function and $d\geq 1$. The degree-$d$ Fourier rank of $f$, denoted $\mathrm{frank}_d(f)$, is the minimal $r$ such that there are polynomials $Q_1,\dots, Q_r$ of degree at most~$d-1$ such that for all $x\in \F_p^n$
	\begin{align}\label{eq:frank_decomposition}
		f(x)=\sum_{i=1}^rc_ie(Q_i(x)).
	\end{align}
\end{definition}
The degree-2 Fourier rank of a function $f$ is also known as the \emph{Fourier sparsity} of $f$. The following lemma relates the notion of rank of a polynomial and its Fourier rank.
\begin{lemma}\label{lemma:rank->frank}
	Let $P\colon \F_p^n\to \T$ be a polynomial and $d\geq 2$. Then $\mathrm{frank}_d(e(P))\geq \mathrm{rank}_d(P)$.
\end{lemma}
\begin{proof}
	Denote by $r'$ the degree $d$ Fourier rank of $e(P)$. So there are polynomials $Q_1,\dots, Q_{r'}$ of degree at most $d-1$ such that we have a decomposition
	\begin{align}\label{eq:P_decomposition}
		e(P(x))=\sum_{i=1}^{r'}c_ie(Q_i(x)).
	\end{align}
	
	We will now define a function $\Gamma\colon \T^{r'}\to \T$ as follows. Let $Q\colon \F_p^{n}\to\T^{r'}$ be defined by~$Q(x)=~(Q_1(x),\dots, Q_{r'}(x))$. The map $\Gamma$ on the image of $Q$ is defined by
	\begin{align*}
		e(\Gamma(Q_1(x),\dots, Q_{r'}(x)))=\sum_{i=1}^{r'}c_ie(Q_i(x)).
	\end{align*}
	For a point $z$ in the complement of the image of $Q$, we simply set $\Gamma(z)=0$ or any other constant value would do. But the $\Gamma$ we just defined has the property that $\Gamma(Q_1(x),\dots, Q_{r'}(x))=P(x)$ by (\ref{eq:P_decomposition}). Hence $r'\geq r$, proving the statement.
\end{proof}
For $d=2$, Sanyal~\cite{sanyal2019fourier} shows that $\mathrm{frank}_2(e(P))\geq \mathrm{rank}_2(P)^2$. This is quadratically better than the above lemma. See Section~\ref{sec:stab_discussion} for a discussion for the case $d>2$.

In the other direction, we have the following lemma.
\begin{lemma}\label{lemma:frank->rank}
	Let $d\geq 2$ and $P\colon \F_p^n\to \T$ a polynomial. Then $$\mathrm{frank}_d(e(P))\leq p^{(1+\lceil(d-1)/(p-1)\rceil)\mathrm{rank}_d(P)}.$$
\end{lemma}
\begin{proof}
	Let $r:=\mathrm{rank}_d(P)$. Then there is a function $\Gamma\colon \T^r\to \T$ and polynomials $Q_1,\dots, Q_r$ of degree at most $d-1$ such that $\Gamma(Q_1(x),\dots,Q_r(x))=P(x)$. The Fourier expansion of $\Gamma$ gives us a degree $d-1$ Fourier expansion of $P$, namely
	\begin{align*}
		e(\Gamma(Q_1(x),\dots,Q_r(x)))&=\sum_{\alpha\in \widehat{G}}\widehat{\Gamma}(\alpha)e(Q_\alpha(x))\\
		&=e(P(x)),
	\end{align*}
	where $Q_\alpha(x)=\sum_{i=1}^r\alpha_iQ_i(x)$ and $\widehat{G}=\prod_{i=1}^r\Z_{p^{k_i+1}}$ is the dual of the group $G=~\prod_{i=1}^r\frac{1}{p^{k_i+1}}\Z/\Z$ and $k_i$ is the depth of $Q_i$.
	Since there are at most	$$|\widehat{G}|\leq~p^{(1+\lceil(d-1)/(p-1)\rceil)r}$$ Fourier coefficients in the above expression, the result follows.
\end{proof}

\section{Magic states in prime dimension}
In this section we will give the explicit form of the magic state that we obtain by using the generalization of the $T$ gate in odd prime dimensions \cite{howard2012qudit}. We will then proceed to show that these magic states have exponentially small correlation with quadratic phase functions. 

\subsection{Generalization of the $T$ gate}\label{subsec:magic_states}
The generalization of the $T$ gate, and hence the corresponding magic state, curiously enough depends on whether the prime dimension $p$ is equal to three or $p>3$. 
\begin{itemize}
	\item But first, let us consider the case $p=2$. In this case, the $T$-gate is given by $T = \begin{psmallmatrix}
		1 & 0\\
		0& e^{i\pi/4}
	\end{psmallmatrix}$ and the corresponding single qubit magic state is given by $T\ket{+}=\ket{T} = \frac{\ket{0}+e^{i\pi/4}\ket{1}}{\sqrt{2}}$, hence the $n$-qubit magic state is
	\begin{equation}\label{eq:magicstate_p_2}
		\ket{T}^{\otimes n} = \frac{1}{2^{n/2}}\sum_{x\in \F_2^n}e(|x|/8)\ket{x}.
	\end{equation}
	Here, the polynomial $P\colon \F_2^n\to \T\colon x\mapsto |x|/8$ is a nonclassical polynomial of degree exactly three. This follows immediately from Proposition \ref{prop:global_nonclassical_poly}, see also Example \ref{example:cubic_polynomials}. We will see a similar phenomenon in other prime dimensions.
	\item In the case that $p=3$, let $\xi =  e^{2\pi i/9}$ be a ninth-root of unity. The generalization of the~$T$-gate for $p=3$, denoted by $U$, is given by
	\begin{equation*}
		U = \begin{pmatrix}
			1 & 0 & 0\\
			0 & \xi & 0 \\
			0 &0 & \xi^2
		\end{pmatrix}.
	\end{equation*}
	The corresponding single qudit magic state is then
	\begin{equation*}
		\ket{\psi_U} = U\ket{+} = \frac{\ket{0}+\xi\ket{1}+\xi^2\ket{3}}{\sqrt{3}}.
	\end{equation*}
	The $n$-qudit magic state in this case is
	\begin{equation}\label{eq:magicstate_p3}
		\ket{\psi_U}^{\otimes n} = \frac{1}{3^{n/2}}\sum_{x\in \F_3^n}e(|x|/9)\ket{x}.
	\end{equation}
	The polynomial $P\colon \F_3^n\to\T\colon x\mapsto |x|/9$ is again a nonclassical polynomial of degree three which follows from Proposition \ref{prop:global_nonclassical_poly}, see also Example \ref{example:cubic_polynomials}.
	
	\item The case $p>3$ are all similar, but different from the previous two cases. In this case, let~$\omega_p=e^{2\pi i/p}$ be a $p$-th root of unity and $P\colon \F_p\to\F_p$ a classical polynomial of degree three. Define 
	\begin{equation*}
		U:= \sum_{x\in \F_p}\omega_p^{P(x)}\ketbra{x}{x}.
	\end{equation*}
	This gate is then a non-clifford gate \cite{howard2012qudit} and could serve as a generalization of the $T$ gate. The condition that $P$ has degree three is important: if $P$ has degree two, then~$U$ is a Clifford gate. The corresponding $n$-qudit magic state is then
	\begin{equation}\label{eq:magicstate_p5}
		\ket{\psi_U}^{\otimes n} = \frac{1}{p^{n/2}}\sum_{x\in \F_p^n}\omega_p^{P_n(x)}\ket{x},
	\end{equation}
	where $P_n(x)= \sum_{i=1}^n P(x_i)$.  Unlike the $p=2,3$ case, the polynomial $P_n$ is a \emph{classical} polynomial of degree three. Coincidentally, there are no cubic nonclassical polynomials for~$p>3$.
\end{itemize}

\subsection{Correlation with quadratic phase functions}
We will now show that for the $n$-qudit magic states defined as above, the correlation with quantum states whose amplitudes are given by a quadratic phase functions is exponentially small.
We need the following basic lemma, see \cite{lovett2008inverse}. We will state and prove it here for convenience.
\begin{lemma}\label{lemma:lovett}
	For any two functions $f,g\colon \F_p^n\to \T$, we have
	\begin{align*}
		|\langle e(f),e(g)\rangle|^4\leq \E_{h\in \F_p^n}|\langle e(\Delta_h f), e(\Delta_h g) \rangle|^2.
	\end{align*}
\end{lemma}
\begin{proof} By the Cauchy-Schwarz inequality,
	\begin{align*}
		\sqrt{\E_h|\langle e(\Delta_hf),e(\Delta_hg)\rangle|^2} & \geq |\E_h\langle e(\Delta_hf),e(\Delta_hg) \rangle|\\
		&=|\E_{x,h}e(f(x+h)-f(x)-(g(x+h)-g(x)))|\\
		&=|\langle e(f),e(g)\rangle |^2.
	\end{align*}
\end{proof}
In other words, we can compute the correlation between two phase functions by computing the correlations between their derivatives and taking their average. If the derivatives are easier to work with, this will become useful.
\begin{proposition}\label{prop:correlation_Tn_with_quadratic}
	Let $P\colon \F_p^n\to\T$ be the polynomial in the phase of the amplitudes in either Equation (\ref{eq:magicstate_p_2}), (\ref{eq:magicstate_p3}) or (\ref{eq:magicstate_p5}). Then, for any nonclassical polynomial $Q\colon\F_p^n\to~\T$ of degree at most two,
	\begin{align*}
		|\langle e(P),e(Q)\rangle|\leq 2^{-cn},
	\end{align*}
	for some $c>0$ depending on $p$.
\end{proposition}
\begin{proof}
We will only prove it for the case $p=2$, since the other case $p>2$ is very similar. We will show that the polynomial phase function $e(|x|/8)$, giving the amplitudes of the $n$-qubit magic state (see (\ref{eq:magicstate_p_2})), has exponentially small correlation with quadratic phase functions. First we compute the derivative of $P$; let $h\in \F_2^n$. Then
 	\begin{align*}
 		P(x+h)-P(x)&=|x+h|/8-|x|/8=\sum_i |x_i+h_i|/8-|x|/8\\
 		&=\sum_i(|x_i|+|h_i|-2|x_i||h_i|)/8-|x|/8\\
 		&=|h|/8-|x\circ h|/4.
 	\end{align*}
 In the second line, we used the property (\ref{eq:property_||}).
 	Let $f\colon x\mapsto e(|x\circ h|/4)$. We need the magnitudes of the Fourier coefficients of $f$; we have for~$\alpha\in \F_2^n$
 	
 	\begin{align*}
 		\hat{f}(\alpha) = \E_x e^{i\pi (|x\circ h|-2|\alpha\circ x|)/2},
 	\end{align*}
 	so
 	\begin{align*}
 		|\hat{f}(\alpha)|^2&=\E_{x,y}e^{i\pi(|x\circ h|-|y\circ h|-2|\alpha \circ x|+2|\alpha \circ y|)/2}\\
 		&= \E_{y,z}e^{i\pi (|y\circ h|+|z\circ h|-2|y\circ z\circ h| -|y\circ  h|-2|\alpha \circ z|)/2}\\
 		&=\E_{y,z}e^{i\pi (|z\circ h|-2|y\circ z\circ h|-2|\alpha \circ z|)/2}\\
 		&=\E_ze^{i\pi(|z\circ h|-2|\alpha\circ z|)/2}\E_y(-1)^{|y\circ z\circ h|}.
 	\end{align*}
 In the second line we used a change of variables $x=y+z$ and the property in (\ref{eq:property_||}) so that
 \begin{align*}
 	|x\circ h|& = |(y+z)\circ h| = |y\circ h + z\circ h| = |y\circ h|+|z\circ h|-2|(y\circ h) \circ (z\circ h)| \\
 	&= |y\circ h|+|z\circ h|-2|y\circ z\circ h|.
 \end{align*}
 	The expectation over $y$ is 0 unless $h_i=1\Rightarrow z_i=0$ in which case it is equal to 1. So continuing where we left off
 	\begin{align*}
 		&=\frac{1}{2^n}\sum_{z:h_i=1\Rightarrow z_i=0}e^{i\pi(|z\circ h|-2|\alpha\circ z|)/2}\\
 		&=\frac{1}{2^n}\sum_{z:h_i=1\Rightarrow z_i=0}(-1)^{|\alpha\circ z|}\\
 	\end{align*}
 	which is 0 unless $h_i=0\Rightarrow \alpha_i=0$, in which case it equals $2^{-|h|}$. So the non-zero Fourier coefficients of $f$ are at those $\alpha\in \F_2^n$ for which it holds that $h_i=0\Rightarrow \alpha_i=0$. 
 	
 	We will now show that there is exponentially small correlation between $e(P)$ and any nonclassical polynomial phase function $e(Q)$ of degree two. By Lemma \ref{lemma:lovett} and Parseval, we have that 
 	
 	\begin{align*}
 		|\langle e(P),e(Q)\rangle|^4 & \leq \E_h|\langle e(\Delta_h P), e(\Delta_h Q) \rangle|^2\\
 		&= \E_h|\langle \widehat{e(\Delta_h P)}, \widehat{e(\Delta_h Q)} \rangle|^2\\
 		&\leq \E_h 2^{-|h|}\\
 		&=\frac{1}{2^n}\sum_{k=0}^n2^{-k}|\{h\in \F_2^n\colon |h|=k\}|\\
 		&=\left(\frac{3}{4}\right)^n.
 	\end{align*}
 	In the third line, we used that $\Delta_h Q$ is a degree one classical polynomial (using Lemma \ref{lem:nonclassical_degree_1_are_classical}), which means that $e(\Delta_h Q)$ has only one non-zero Fourier coefficient and that the magnitude of all the Fourier coefficients of $e(\Delta_h(P))$ is at most $2^{-|h|/2}$, see the computation of $|\hat{f}(\alpha)|^2$ above where $f(x) = e(|h|/8)e(-\Delta_hP(x))$. 	
 	
For $p=3$, the cubic phase function giving the amplitudes of the $n$-qudit magic states is given by (\ref{eq:magicstate_p3}). One follows the same strategy by first computing the derivative of the polynomial $x\mapsto |x|/9$ using the property (\ref{eq:property_||_3}) and computing the magnitude of the Fourier coefficients of its associated phase function and concluding with a use of Lemma~\ref{lemma:lovett}. The case $p>3$ is simpler as it does not require any property of the map $|\cdot|$, since the polynomial in the phase of (\ref{eq:magicstate_p5}) is a classical polynomial.
\end{proof}

\section{Stabilizer rank of the $n$-qudit magic state}
Here we prove Theorem \ref{thm:main}, our main result.

The following claim from \cite{peleg2021lower} is needed to get a handle on the affine subspaces that appear with the stabilizer states in a stabilizer decomposition.
\begin{claim}[\cite{peleg2021lower}]\label{lem:large_affine_subspace}
	Let $p$ be a prime and $H_1,\dots, H_r\subset \F_p^n$ be a collection of affine subspaces and assume~$r\leq n/2$. Then there exists an affine subspace $U$ of dimension at least $n-2r$ and a subset $S\subset [r]$ such that for all $x\in U$
		\begin{equation*}
			1_{H_i}(x)=
			\begin{cases*}
				1 & if $i\in S$ \\
				0 & otherwise.	
			\end{cases*}
		\end{equation*}
	
\end{claim}
\begin{proof}
	Let $E\colon \F_p^n\to \{0,1\}^r$ be the map given by $x\mapsto (1_{H_1}(x),\dots, 1_{H_r}(x))$.
	By the pigeonhole principle, there is $\alpha\in \{0,1\}^r$ such that $|E^{-1}(\alpha)|\geq p^n2^{-r} \geq p^{n-r}$. Denote by $S\subset [r]$ the set of indices~$i$ such that $\alpha_i=1$. It is clear that $E^{-1}(\alpha) = \cap_{i\in S}H_i\setminus \cup_{i\notin S}H_i$ (if $S=\emptyset$, define $\cap_{i\in S}H_i=\F_p^n$). Now define $V = \cap_{i\in S}H_i$, this is an affine subspace of dimension at least $n-r$. Pick an arbitrary $x_0\in E^{-1}(\alpha)$, so~$x_0\notin H_i$ for~$i\notin S$. This implies that $\forall i\notin S$ there is an affine equation $h_i$ such that $h_i(x_0)=1$ but~$h_i(x)=0$ for all $x\in H_i$. The affine subspace we are looking for is
	\begin{equation*}
		U:=\{x\in V\colon \forall i\notin S\quad h_i(x)=1\}. 
	\end{equation*}
	Note that $U$ is not empty since $x_0\in U$. Since we only add at most $r$ extra equations, the dimension of $U$ is at least $n-2r$. 
\end{proof}

\paragraph{Proof of Theorem \ref{thm:main}}
	Let $P\colon \F_p^n\to\T$ be the nonclassical polynomial of degree three given by the corresponding $n$-qudit magic state, depending on the prime $p$. For $p=2$, $p=3$ and $p>3$ it is given by the (polynomials in the phase of the) amplitudes of Equation (\ref{eq:magicstate_p_2}), (\ref{eq:magicstate_p3}) and (\ref{eq:magicstate_p5}) respectively. 
	Let $\ket{\psi}$ be the corresponding $n$-qudit magic state, i.e.
	\begin{align*}
		\ket{\psi}= p^{-n/2}\sum_{x\in \F_p^n}e(P(x))\ket{x}.
	\end{align*}
	Denote by $r$ the stabilizer rank of $\ket{\psi}$, so there is a decomposition
	\begin{align*}
		\ket{\psi}=\sum_{i=1}^{r}c_i\ket{\phi_i},
	\end{align*}
	for some constants $c_i$ and stabilizer states $\ket{\phi_i}$. Each such $\phi_i$ is defined on an affine subspace~$H_i\subset \F_p^n$. Let $C>0$ be a large enough constant\footnote{Given the prime $p$ and the constant $c$ from Proposition \ref{prop:correlation_Tn_with_quadratic}, the constant $C$ should satisfy $C>2p/c$.}. If $r>n/C$, we are done. So assume~$r\leq n/C$.
	Then, we have that
	
	\begin{align*}
		e(P(x))&=p^{n/2}\braket{x}{\psi}=p^{n/2}\sum_{i=1}^rc_i\braket{x}{\phi_i}\\
		&=\sum_{i=1}^rc'_ie(Q_i(x))1_{H_i}(x),
	\end{align*}
	where $Q_i$ is a quadratic polynomial on $H_i$ and $c_i'=p^{(n-\dim(H_i))/2}c_i$.
Then by Claim \ref{lem:large_affine_subspace} (using~$r\leq~n/C$), there is an affine subspace $U$ of dimension $cn$ for some $c\geq 0.99$ and a non-empty subset $S\subset [r]$ such that $\forall x\in U$

		\begin{equation*}
	1_{H_i}(x)=
	\begin{cases*}
		1 & if $i\in S$ \\
		0 & otherwise.	
	\end{cases*}
\end{equation*}

	 Let $A\colon \F_p^{cn}\to\F_p^n$ be an affine linear map such that $U=\{A(y)\colon y\in\F_p^{cn}\}$. Let $P'\colon\F_p^{cn}\to\T$ be the polynomial given by $P'(y) = P(A(y))$. Then the restriction of the above decomposition of $e(P(x))$ to $U$ implies that
	\begin{align}\label{eq:P_restricted}
		e(P'(y))=\sum_{i\in S}c'_ie(Q'_i(y)),
	\end{align}
	where $Q'_i(y):=Q_i(A(y))$ and we have that $Q'_i$ is a nonclassical polynomial of degree at most two. By Proposition \ref{prop:correlation_Tn_with_quadratic} and (the contrapositive of) \ref{prop:rank_implies_uniform} the polynomial $P$ has high rank: $\mathrm{rank}(P)\geq \Omega(n)$. By Lemma \ref{lem:restrict_pol_subspace]} it follows that $P'$ is still cubic and we have $\mathrm{rank}(P')\geq \Omega(n)$ (this uses that $C$ is a large enough constant).
	But \eqref{eq:P_restricted} is a decomposition in terms of nonclassical polynomial phase functions of degree at most two. By Lemma~\ref{lemma:rank->frank} we have $|S|\geq \Omega(n)$ so that $r \geq \Omega(n)$.\qed

\section{Discussion}\label{sec:stab_discussion}
From the above proof, we can immediately conclude that it is not possible to get super-linear lower bounds on the stabilizer rank of $n$-qudit magic states using these techniques. This is due to the use of Claim~\ref{lem:large_affine_subspace}.
However, there is no obvious obstruction to get super-linear lower bounds on the number of stabilizer states needed in a decomposition of $n$-qudit magic states where all the stabilizer states are defined on the full space $\F_p^n$. The graph states (classical quadratic polynomials) are for example in this set. The possibility of such a super-linear lower bound hinges on the relationship between the rank of a polynomial and its Fourier rank: the $d$-rank of a polynomial~$P$ on $n$ variables is at most $n$, whereas the degree-$d$ Fourier rank of a polynomial is at most $p^n$. Lemma~\ref{lemma:rank->frank} only shows that $\mathrm{frank}_d(e(P))\geq\mathrm{rank}_d(P)$. Is this relation optimal, or can we expect much better?
\begin{problem}\label{prob:rank_frank}
	Let $d\geq 2$ and $P\colon \F_p^n\to \T$ be a polynomial. Is it true that
	\begin{align*}
		\mathrm{frank}_d(e(P))\geq \omega(\mathrm{rank}_d(P))?
	\end{align*}
\end{problem}
A positive answer to this question would not only show that the $n$-qubit magic state needs super-linear many stabilizer states defined on the full space $\F_p^n$, but would also have implications in another field. For this, let us consider the case $p=2$.

Let $\mathrm{AND}(x) = |x_1x_2\cdots x_n|/2$ be the (classical) polynomial giving the AND function. The ``quadratic uncertainty principle'' \cite{filmus2014real} is a conjecture that states that any decomposition
\begin{align}\label{eq:qup}
	e(\mathrm{AND}(x)) = \sum_{i=1}^r c_ie(Q_i(x)),
\end{align}
where $Q_i$ are classical quadratic polynomials, must have $r\geq 2^{\Omega(n)}$. Note that Definitions \ref{def:rank_poly} and \ref{def:frank} of rank and frank are still valid if we only allow \emph{classical} polynomials, which we will denote by $\mathrm{rank}'$ and $\mathrm{frank}'$. The best known lower bound is $\mathrm{frank}'_3(e(\mathrm{AND}))\geq n/2$ \cite{williams:LIPIcs:2018:8884}. The proof of this uses the Chevalley-Warning theorem in an elegant way. It shows, implicitly, that the 3-$\mathrm{rank}'$ of the AND function is at least $n/2$, i.e. near maximal rank. This way of looking at it fits very well with Lemma \ref{lemma:rank->frank}. What is actually shown in \cite{williams:LIPIcs:2018:8884} is the same lower bound for the NOR function, i.e. $\mathrm{NOR}(x) = |1+x_1|\cdots |1+x_n|/2$. Lower bounds on the NOR function imply the same lower bounds on the AND function (it is the same function in a different basis). The proof also works if we allow polynomials of degree at most a constant.

\begin{theorem}[\cite{williams:LIPIcs:2018:8884}, generalized]
	Let $\mathrm{NOR}\colon \F_2^n\to\T$ be the function defined above and let~$d\geq 3$ be an integer (constant). Then $\mathrm{frank}'_d(e(\mathrm{NOR}))\geq n/(d-1)$.
\end{theorem}
\begin{proof}
	We will show that $\mathrm{rank}'_d(\mathrm{NOR})\geq n/(d-1)$. Having shown this, the result follows immediately from Lemma \ref{lemma:rank->frank} (which also holds for $\mathrm{rank}'$ and $\mathrm{frank}'$).
	
	Let $r=\mathrm{rank}'_d(\mathrm{NOR})$. So there are (classical) polynomials $Q_1,\dots,Q_r$ of degree at most $d-1$ such that there is a function $\Gamma\colon\T^r\to\T$ such that
	\begin{align*}
		\mathrm{NOR}(x) = \Gamma(Q_1(x),\dots, Q_r(x)).
	\end{align*}
	We may assume without loss of generality that $Q_i(0)=0$ for all $i=1,\dots , r$ (since $\Gamma$ can always add/subtract a constant to/from $Q_i$). Assume that $r<n/(d-1)$. By the Chevalley-Warning theorem, the polynomials $Q_1,\dots, Q_r$ have another common root $x\neq (0,\dots, 0)$. This contradicts the fact that $\mathrm{NOR}(x)=1/2$ if and only if $x=(0,\dots,0)$. 
\end{proof}

Since the NOR function has $\mathrm{rank}'_3(\mathrm{NOR})\geq n/2$, a positive answer to Problem \ref{prob:rank_frank} would show that the number of quadratic polynomials needed in (\ref{eq:qup}) is super-linear. As noted before, Sanyal~\cite{sanyal2019fourier} showed that $\mathrm{frank}'_2(e(P))\geq\mathrm{rank}'_2(P)^2$. An analogue of this result for $d>2$ would quadratically improve the best lower bound on $\mathrm{frank}'_3(\mathrm{NOR})$.

\bibliographystyle{plainnat}
\bibliography{stab_prime.bib}

\end{document}